\documentclass{article}

\usepackage[letterpaper]{geometry}

\usepackage{amsmath, amssymb, etoolbox}

\usepackage[noline,noend,linesnumberedhidden,boxed,algo2e]{algorithm2e}
\DontPrintSemicolon
\SetArgSty{textnormal}
\SetAlCapSkip{1ex}
\newtheorem{theorem}{Theorem}
\newtheorem{lemma}[theorem]{Lemma}

\newcommand{\qed}{\rule{7pt}{7pt}}

\newenvironment{proof}{\noindent {\it Proof\/}:}{\qed \medskip}

\newcommand{\argmin}{\hbox{argmin\,}}

\begin{document}

\bibliographystyle{abbrv}

\title{Recursive Random Contraction Revisited}

\author{David R.\ Karger\thanks{Address: MIT, Computer Science and Artificial Intelligence Laboratory, The Stata Center, Room G592, 32 Vassar St., Cambridge, MA 02139, USA. {\tt karger@mit.edu}}\\MIT \and
	David P.\ Williamson\thanks{Address: School
		of Operations Research and Information Engineering, Cornell
		University, Ithaca, NY 14853, USA. {\tt davidpwilliamson@cornell.edu}}\\Cornell University}

%% CHANGED
%\thispagestyle{empty}

\date{}
\maketitle

% Default Copyright Statement
%\fancyfoot[R]{\scriptsize{Copyright \textcopyright\ 2021 by SIAM\\
%		Unauthorized reproduction of this article is prohibited}}

\begin{abstract}
	
In this note, we revisit the recursive random contraction algorithm of Karger and Stein \cite{KargerS96} for finding a minimum cut in a graph.  Our revisit is occasioned by a paper of Fox, Panigrahi, and Zhang \cite{FoxPZ19} which gives an extension of the Karger-Stein algorithm to minimum cuts and minimum $k$-cuts in hypergraphs.  When specialized to the case of graphs, the algorithm is somewhat different than the original Karger-Stein algorithm.  We show that the analysis becomes particularly clean in this case: we can prove that the probability that a fixed minimum cut in an $n$ node graph is returned by the algorithm is bounded below by $1/(2H_n-2)$, where $H_n$ is the $n$th harmonic number.  We also consider other similar variants of the algorithm, and show that no such algorithm can achieve an asymptotically better probability of finding a fixed minimum cut.
\end{abstract}

%% CHANGE
%%
%\setcounter{page}{0}
%\newpage
\section{Introduction}
In the global minimum cut problem (or simply the minimum cut problem), we are given an undirected graph $G=(V,E)$, with a capacity $u(i,j) \geq 0$ for each edge $(i,j) \in E$.  The goal is to find a nontrivial set $S \subset V$, $S \neq \emptyset$, that minimizes the capacity of the cut $S$: the capacity is the sum of the capacities of the edges with exactly one endpoint in $S$.  Let $\delta(S)$ be the set of edges with exactly one endpoint in $S$.  We denote the capacity of the cut $S$ by $u(\delta(S))$.  We wish to find a nontrivial $S$ that minimizes $u(\delta(S))$.

It has long been known that it is possible to find a minimum cut via $n-1$ maximum flow computations (pick an arbitrary source $s$, replace each undirected edge $(i,j)$ with two arcs $(i,j)$ and $(j,i)$ of capacity $u(i,j)$, for all other $t \neq s$ find the minimum $s$-$t$ cut, and return the smallest cut found).  In the 1990s, a number of different, non-flow-based algorithms were proposed to find the minimum cut, including ones proposed by Gabow \cite{Gabow95} and Nagamochi and Ibaraki \cite{NagamochiI92}.  Karger \cite{Karger93} proposed a particularly simple randomized algorithm: pick a random edge with probability proportional to its capacity, and {\em contract} the edge; that is, identify its two endpoints.  The basic idea is that an edge chosen with probability proportional to its capacity is unlikely to be in $\delta(S^*)$ for a given minimum cut $S^*$.  We say that the cut $S^*$ {\em survives} the contraction if the contracted edge is not in $\delta(S^*)$.  In particular, Karger shows that the probability that a given minimum cut $S^*$ survives the contraction for an edge chosen with probability proportional to its capacity is at least $1- \frac{2}{n}$ for an $n$-node graph.  The run time of Karger's algorithm is $O(m)$; it finds a given minimum cut with probability at least $1/{n \choose 2}$.  By running the algorithm many times, the algorithm can find a given minimum cut with high probability in $O(n^4 \log n)$ time.

Karger and Stein \cite{KargerS96} improve Karger's original contraction algorithm by showing that if one contracts an $n$ node graph to have $n/\sqrt{2}$ nodes, the probability that a given minimum cut survives is at least 1/2.  Thus they do this twice, call their  algorithm recursively, and return the better of the two cuts found. Then we expect in at least one of the two recursive calls, the given minimum cut survives. They prove that the probability that this algorithm returns a given minimum cut is $\Omega(1/\log n)$, much better than the original algorithm.  The algorithm runs in $O(n^2 \log n)$ time.  By running the algorithm $O(\log^2 n)$ times, it can return a given minimum cut with high probability.

\iffalse
\begin{algorithm}[t]
	\TitleOfAlgo{RecursiveContraction($G$,$n$)}
	\eIf{$n \leq 6$}{
		Find global minimum cut in $G$ by exhaustive search\;
	}{
		\For{$i \gets$ 1 to 2}{
			$H_i \gets$ random contraction of $G$ down to $\lceil n/\sqrt{2} + 1\rceil$ vertices\;
			$S_i \gets$ \ProcNameSty{RecursiveContraction}($H_i$, $\lceil n/\sqrt{2} + 1\rceil$)\;
		}
		$j \gets \argmin_{j=1,2} u(\delta(S_j))$\;
		\Return($S_j$)
	}
\caption{The Karger-Stein recursive random contraction algorithm.}
\label{alg:ks}
\end{algorithm}
\fi

Recently, Fox, Panigrahi, and Zhang \cite{FoxPZ19} proposed an extension of the Karger-Stein recursive contraction algorithm to finding minimum cuts in hypergraphs.  When viewed as a recursive contraction algorithm on graphs, the Fox, Panigrahi, and Zhang version of the algorithm is somewhat different than the original Karger-Stein algorithm.  
One can give two different perspectives on the algorithm.  The first is that 
on an $n$ node graph the algorithm draws a number of recursive calls $k$ from a geometric distribution with probability $p_n = 1 - \frac{2}{n}$, so that the probability of $k$ recursive calls is $p_n(1-p_n)^{k-1}$. For each call, the algorithm picks a random edge to contract, contracts it, then recursively calls itself on the contracted graph.
Alternatively, the algorithm picks a random edge, contracts it, and calls itself on the smaller graph; let $S_1$ be the minimum cut found on this call. With probability $1-p_n$ it additionally calls itself on the same graph, resulting in cut $S_2$, and returns the smaller of cuts $S_1$ and $S_2$; otherwise it returns $S_1$. We can view this second algorithm as flipping a coin repeatedly until we obtain a heads, with the probability of heads being $p_n$:
for each coin flip we contract an edge and perform a recursive call. We will perform $k$ calls with probability $(1-p_n)^{k-1}p_n$, just as in the first version of the algorithm. We present the two different perspectives in Algorithms \ref{alg:fpz1} and \ref{alg:fpz2}.  
We call this algorithm the FPZ recursive contraction algorithm, or just the FPZ algorithm, and we summarize the two variants in Algorithm \ref{alg:fpz1} and Algorithm \ref{alg:fpz2}.

\begin{algorithm2e}[t]
	\TitleOfAlgo{FPZRecursiveContraction($G$,$n$)}
	\eIf{$n = 2$}{
		\Return one of the two nodes\;
	}{
		Pick $k$ with probability $p_n(1-p_n)^{k-1}$\;
		\For{$i \gets$ 1 to $k$}{
			Let $G_i$ be result of contracting one edge chosen with probability proportional to capacity\;
			$S_i \gets$ \ProcNameSty{FPZRecursiveContraction}($G_i$, $n-1$)\;
		}
		$j \gets \argmin_{j=1,\ldots,k} u(\delta(S_j))$\;
		\Return($S_j$)
	}
\caption{The Fox-Panigrahi-Zhang recursive contraction algorithm.}
\label{alg:fpz1}
\end{algorithm2e}

\begin{algorithm2e}[t]
	\TitleOfAlgo{FPZRecursiveContraction($G$,$n$)}
	\eIf{$n = 2$}{
		\Return one of the two nodes\;
	}{
		Pick random number $r$ uniformly in $[0,1]$\;
		Let $G'$ be result of contracting one edge chosen with probability proportional to capacity\;
		$S_1 \gets $\ProcNameSty{FPZRecursiveContraction}($G'$,$n-1$)\;
		\eIf{$r \leq p_n$}{
			\Return($S_1$)\;
		}{
			$S_2 \gets $\ProcNameSty{FPZRecursiveContraction}($G$,$n$)\;
			$j \gets \argmin_{j=1,2} u(\delta(S_j))$\;
			\Return($S_j$)
		}
	}
	\caption{The Fox-Panigrahi-Zhang recursive contraction algorithm (version 2).}
	\label{alg:fpz2}
\end{algorithm2e}

We show that the analysis becomes particularly clean for the FPZ algorithm applied to graphs: we can provide an exact expression that lower bounds the probability that a given minimum cut is returned.  In particular, we show that the probability that a given minimum cut is returned for an $n$-node graph is bounded below by $1/(2H_n-2)$, where $H_n$ is the $n$th harmonic number; that is, $H_n = 1 + \frac{1}{2} + \frac{1}{3} + \cdots + \frac{1}{n}$.  The lower bound given by FPZ \cite[Theorem 2.1]{FoxPZ19} for finding a minimum cut in hypergraphs reduces to this quantity, and they note that the bound is tight in the case of cycles in graphs.  $H_n$ is approximately $\ln n$, so that the probability of the algorithm returning a given minimum cut is $\Omega(1/\log n)$, the same as that of the original Karger-Stein algorithm.   Fox et al.\ show that the running time of the algorithm for graphs is $O(n^2 \log n)$, and we give a different (and perhaps simpler) proof.  We note that unlike the Karger-Stein analysis, we do not need to analyze the probability that a given minimum cut survives multiple contractions; it is sufficient to know that the probability that it survives a single contraction is at least $1- \frac{2}{n}$ for an $n$-node graph.

We then take the occasion to consider other variants of the recursive contraction algorithm.  In particular, we show that for a class of recursive contraction algorithms, we may as well assume that the algorithm performs just one contraction and then makes one recursive call with probability $(n-4)/(n-2)$ and two recursive calls with probability $2/(n-2)$.  This algorithm also succeeds in finding a given minimum cut with probability $\Omega(1/\log n)$ and has running time $O(n^2 \log n)$.  In fact, we show that any algorithm from this class that runs in $O(n^2\log n)$ time has success probability $\Theta(1/\log n)$, so that no such algorithm can achieve an asymptotically better probability of finding a fixed minimum cut.

\section{The FPZ Algorithm}

In order to make our presentation self-contained, we show the lemma from Karger \cite{Karger93} that the probability of a given minimum cut surviving a random contraction is at least $p_n \equiv 1 - \frac{2}{n}$.

\begin{lemma}[Karger \cite{Karger93}]
Let $S^*$ be a given minimum cut, and let $\lambda^*  = u(\delta(S^*))$ be the capacity of the cut.  If we choose an edge $(i,j)$ at random with probability proportional to its capacity $u(i,j)$, then the probability that $(i,j) \notin \delta(S^*)$ is at least $p_n$.
\end{lemma}

\begin{proof}
Let $U$ be the total capacity of all edges in the graph, so that $U=\sum_{(i,j) \in E} u(i,j)$.  The probability we choose a given edge $(i,j)$ is $u(i,j)/U$, and the probability edge $(i,j) \notin \delta(S^*)$ is $$\frac{U - \lambda^*}{U} = 1 - \frac{\lambda^*}{U}.$$  We observe that $u(\delta(i)) \geq \lambda^*$ for any $i \in V$ since $\lambda^*$ is the capacity of the minimum cut.  Then $$U = \frac{1}{2} \sum_{i \in V} u(\delta(i)) \geq \frac{n}{2} \lambda^*.$$  Thus the probability is at least
$$1 - \frac{\lambda^*}{n\lambda^*/2} = 1 - \frac{2}{n} = p_n.$$
\end{proof}

Initially, the choice of recursive calls for the FPZ algorithm seems perplexing.  We show below that the expected number of calls in which a given minimum cut survives is one.  We justify why we need the algorithm to have at least one recursive call in which the given minimum cut survives via {\em branching processes}. This is a model of a tree in which each parent node gives birth some number of child nodes, where the number is drawn from some distribution.  A key theorem is that when the expected number of child nodes is less than 1, then the branching process rapidly goes {\em extinct}; that is, at some depth of the tree there are no more child nodes.  Conversely, if the expected number of children exceeds 1, then there is a good chance for the tree to have infinite depth and to have a number of children at depth $d$ that is exponential in $d$.

We can consider our recursive process as a branching process where the children of a node are those in which the given minimum cut survives.  If this expected number is below 1, then the recursive calls in which the given cut survives die out quickly, so we have little chance of finding the cut at one of the leaves of the recursion.  So we want the expectation to be at least one.  
Similarly (and perhaps a bit surprisingly), we also do not want the expected number of recursive calls in which the given minimum cut survives to be too large.  If it is larger than 1, then at depth $n-1$ (after $n-1$ contractions), the recursion tree will have a number of leaves in which the given cut survives that is exponential in $n$, and thus a total number of leaves in which is exponential in $n$.  This leads to a running time exponential in $n$.

In summary, to achieve a good success probability along with a good running time, we want the expected number of children to be equal to 1, and this is what the FPZ algorithm does. (This is observed in \cite{FoxPZ19} as well, for instance, just before Lemma 2.1).

\begin{lemma}
The expected number of recursive calls of the FPZ algorithm that returns a given minimum cut is at least one.
\end{lemma}

\begin{proof}
To show this, it helps to use the first perspective on the algorithm. For a geometric distribution, the expected number of calls is $\sum_{k=1}^\infty p_n(1-p_n)^{k-1}k = \frac{1}{p_n}$.  Then we have that the expected number of recursive calls such that the given minimum cut survives is at least
$$\sum_{k=1}^\infty \Pr[k \mbox{ recursive calls}] \cdot  E[\mbox{number calls that cut survives}|k \mbox{ calls}] \geq  \sum_{k=1}^\infty p_n(1-p_n)^{k-1} \cdot k \cdot p_n,$$
using that in each of the $k$ calls, the probability that the cut survives the single random contraction is at least $p_n=1 - \frac{2}{n}$, so that if there are $k$ calls, the expected number of calls in which the cut survives is $kp_n$.  Continuing, and using the fact that the expected number of calls is $1/p_n$, the expected number of calls in which the given minimum cut survives is at least
$$p_n \sum_{k=1}^\infty p_n(1-p_n)^{k-1} \cdot k = p_n \cdot \frac{1}{p_n} = 1.$$
\end{proof}

We now prove our main theorem.  We fix a particular global minimum cut $S^*$.  This result is implied by FPZ \cite[Theorem 2.1]{FoxPZ19}.

\begin{theorem} \label{thm1}
The probability that the given global minimum cut $S^*$ is returned by the FPZ algorithm is bounded below by $$\frac{1}{2H_n-2}.$$
\end{theorem}

\begin{proof}
Let $Q(n)$ denote the probability that $S^*$ is returned by the algorithm, given that it has survived contractions to $n$ nodes, and given that the probability $S^*$ survives a contraction from $k$ nodes to $k-1$ nodes is exactly $p_k$.  Then $Q(n)$ gives a lower bound on the probability that $S^*$ will survive.  

We now give a recursive equation for $Q(n)$, using the second perspective on the FPZ algorithm (Algorithm \ref{alg:fpz2}). 
We have $Q(2) = 1$, and we set $Q(n)$ as follows:
%$$ p_n^2Q(n-1) + (1-p_n)\left(1 - \left(1 - Q(n)\right)\left(1 - p_n \cdot Q(n-1)\right)\right);$$
$$Q(n) =p_n^2Q(n-1) + (1-p_n)\left(1 - \left(1 - Q(n)\right)\left(1 - p_n \cdot Q(n-1)\right)\right).$$
This equation follows since with probability $p_n$ we perform a random contraction and a recursive call on a graph of size $n-1$, and the probability the cut survives is $p_nQ(n-1)$. With probability $(1-p_n)$ we perform two calls, one on the original graph of size $n$ and one on a graph of size $n-1$ with a single edge contracted  and the probability that the cut does not survive in either call is $(1-Q(n))(1-p_nQ(n-1))$, so the probability it survives at least one of the two calls is $1-(1-Q(n))(1-p_nQ(n-1))$.

Rewriting, we have that
\begin{align*}
	Q(n) & =p_n^2Q(n-1) + (1-p_n)\left(1 - \left(1 - Q(n)\right)\left(1 - p_n \cdot Q(n-1)\right)\right) \\
	&  = p_n^2Q(n-1) + (1-p_n)\left(Q(n) + p_nQ(n-1) - p_n Q(n-1)Q(n)\right)  \\
	& =  p_n^2Q(n-1) + p_nQ(n-1) - p_n^2Q(n-1) + (1 - p_n)Q(n) - p_n (1-p_n) Q(n-1)Q(n)  \\
	& = p_nQ(n-1) + \left((1-p_n) - p_n(1-p_n)Q(n-1)\right)Q(n).
\end{align*}
Then we have that
$$p_nQ(n) = p_nQ(n-1) - p_n(1-p_n)Q(n-1)Q(n),$$
or 
$$Q(n) = Q(n-1) - (1-p_n)Q(n-1)Q(n),$$
or, dividing by $Q(n)Q(n-1)$, and recalling that $1-p_n = \frac{2}{n}$,
$$\frac{1}{Q(n-1)} = \frac{1}{Q(n)} - \frac{2}{n},$$ so that
$$
\frac{1}{Q(n)}  = \frac{1}{Q(n-1)} + \frac{2}{n}  = \frac{1}{Q(n-2)} + \frac{2}{n-1} + \frac{2}{n} = \cdots = 2H_n - 3 + \frac{1}{Q(2)}.$$
 Using that $Q(2) = 1$, we get that $$\frac{1}{Q(n)} = 2H_n-2,$$ or $Q(n) = \frac{1}{2H_n - 2}.$
\iffalse
$$Q(n) \left(p_n +  p_n (1-p_n)Q(n-1)\right) = p_nQ(n-1),$$
or $$Q(n) \left(1 + (1-p_n)Q(n-1)\right) = Q(n-1),$$ or
$$Q(n) = \frac{Q(n-1)}{1 + (1-p_n)Q(n-1)}.$$  Guessing that $Q(n) = \frac{1}{2H_n-2}$, recalling that $1-p_n = \frac{2}{n}$, and plugging this into the right-hand side, we get that
\begin{align*}
\frac{Q(n-1)}{1 + \frac{2}{n}Q(n-1)} & = \frac{\frac{1}{2H_{n-1}-2}}{1 + \frac{2}{n}\cdot\frac{1}{2H_{n-1}-2}}\\
& = \frac{1}{2H_{n-1}-2 + \frac{2}{n}}\\
& = \frac{1}{2H_n-2} = Q(n),
\end{align*}
as desired.  It is also easy to check that $Q(2) = 1$.
\fi
\end{proof}

\begin{theorem} \label{thm:runtime}
The expected running time of the algorithm is $O(n^2 \log n)$.
\end{theorem}

\begin{proof}
Let $T(n)$ be the expected running time of the algorithm on an $n$ node graph. Karger and Stein have shown that the time to perform a single random contraction is $O(n)$.  Then using the second perspective on the algorithm (Algorithm \ref{alg:fpz2}), we have that
$$T(n) = T(n-1) + (1-p_n)T(n) + O(n),$$
or
$$p_n T(n) = T(n-1) + O(n),$$
or, since $p_n = \frac{n-2}{n}$, $$\frac{1}{n}T(n) = \frac{1}{n-2}T(n-1) + O\left(\frac{n}{n-2}\right),$$ or
$$
\frac{1}{n(n-1)} T(n) =  \frac{1}{(n-1)(n-2)}T(n-1) + O\left(\frac{n}{(n-1)(n-2)}\right).$$
Substituting $R(n) = \frac{1}{n(n-1)} T(n)$, we get
$$R(n) = R(n-1) + O\left(\frac{1}{n}\right),$$ so that $R(n) = O(H_n)$, and thus $T(n)  = O(n^2 \log n)$.
\end{proof}

\section{Limits on Our Approach}

In this section, we consider whether a similar version of the recursive contraction algorithm could achieve either a higher probability of success or an asymptotically faster running time, and we show, under some assumptions, that the answer is no.

We assume that in each step the algorithm draws a positive integer $k$ from a distribution with probability $a_k$ (so that $\sum_{k=1}^\infty a_k = 1$), and $k$ times it picks a single edge randomly, contracts it, and calls itself on the contracted graph.  %We assume that the support of this distribution is independent of $n$, although the $a_k$ values may depend on $n$.  
We assume that the expected number of calls in which a given minimum cut survives is one, as we justified earlier.

We let $T$ be the time that the algorithm takes to contract the randomly chosen edge and call itself recursively on the smaller graph, so that the expected running time of the algorithm is $\hat T \equiv \sum_{k=1}^\infty a_k \cdot k \cdot T$.  Furthermore, let $f$ be the probability that the algorithm fails to return the given minimum cut in the process of contracting a single random edge and calling itself recursively on the smaller graph.  Thus the probability that the algorithm fails to return the given minimum cut when it calls itself recursively $k$ times is $f^k$, and the probability that the algorithm fails to return the given minimum cut overall is $\hat F \equiv \sum_{k=1}^\infty a_k \cdot f^k$.  

If we consider the points $(kT,f^k)$ for $k =1,\ldots,\infty$, the lower envelope of these points is a piecewise linear curve joining the points. Then  $(\hat T, \hat F)$ can be expressed as a convex combination of the points on the curve, and must lie on or above the lower envelope. Thus there exists some $k$ and some $0 \leq \lambda \leq 1$ such that 
$\hat T = \lambda kT + (1-\lambda)(k+1)T$ and such that $\hat F \geq \lambda f^k + (1-\lambda)f^{k+1}.$  In other words, we can obtain the same expected running time and no greater of a failure probability if we change the probability distribution to perform $k$ recursive calls with probability $\lambda$ and $k+1$ recursive calls with probability $1-\lambda$.

%We observe then that if $k\ge 2$ so the algorithm makes at least two recursive calls each time, then the recursion tree will always %have $\Omega(2^n)$ leaves in any execution, and thus the algorithm will take $\Omega(2^n)$ time.  Thus we must have $k=1$ at each %level, and that we either make one recursive call with probability $\lambda$, or two recursive calls with probability $1 - \lambda$.  

We now use the fact that we want the given minimum cut to survive exactly one of the recursive calls in order to determine $\lambda$ and $k$.  The given cut survives the random contraction of one edge with probability $p_n$, so the expected number of calls in which it survives when we make $k$ calls with probability $\lambda$ and $k+1$ calls with probability $1 - \lambda$ is 
$$\lambda \cdot k \cdot p_n + (1-\lambda)\cdot(k+1) \cdot p_n.$$  Setting this quantity equal to one, we get that
$$\lambda \cdot k \cdot p_n + (1-\lambda)\cdot (k+1) \cdot p_n=1,$$ or 
$$1 + \lambda \cdot p_n = (k+1) \cdot p_n,$$
or $$\lambda = (k+1) - \frac{1}{p_n} = k+1 - \frac{n}{n-2}.$$  Note then that for $n \geq 5$ that $0 \leq \lambda \leq 1$ only if $k=1$.  Then for $k=1$,
$$\lambda = 2 - \frac{1}{p_n} = 2 - \frac{n}{n-2} = (n-4)/(n-2).$$
Thus we have shown that the optimal choice is that the algorithm makes one recursive call with probability $(n-4)/(n-2)$ and two recursive calls with probability $1 - \frac{n-4}{n-2} = 2/(n-2)$.

We can now bound the probability that the algorithm successfully returns a given minimum cut.  As in the previous section, let $Q(n)$ denote the probability that $S^*$ is returned by the algorithm, given that it has survived contractions to $n$ nodes, and given that the probability $S^*$ survives a contraction from $k$ nodes to $k-1$ nodes is exactly $p_k$.

\begin{theorem}
$Q(n) = \Theta(1/\log n)$.
\end{theorem}

\begin{proof}
We have that
\begin{eqnarray*}
 Q(n) & = & \frac{n-4}{n-2} \cdot p_n \cdot Q(n-1) + \frac{2}{n-2} \cdot (1 - (1 - p_nQ(n-1))^2)\\
& = & \frac{n-4}{n-2} \cdot p_n \cdot Q(n-1) + \frac{2}{n-2}(2p_nQ(n-1)-p_n^2Q(n-1)^2)\\
& = & \frac{n}{n-2} \cdot p_n \cdot Q(n-1) - \frac{2}{n-2}p_n^2Q(n-1)^2\\
& = & \frac{n}{n-2} \cdot \frac{n-2}{n}Q(n-1) - \frac{2}{n-2}\frac{n-2}{n}\frac{n-2}{n} Q(n-1)^2\\
& = & Q(n-1) - \frac{2(n-2)}{n^2}Q(n-1)^2.
\end{eqnarray*}
We substitute $S(n) = 1/Q(n)$, so that we have
$$\frac{1}{S(n)} = \frac{1}{S(n-1)} - \frac{2(n-2)}{n^2}\frac{1}{S(n-1)^2}.$$
Multiplying by $S(n)S(n-1)$ and rearranging, we get 
$$S(n) - S(n-1) = \frac{2(n-2)}{n^2}\frac{S(n)}{S(n-1)}.$$
Since $Q(n) \leq Q(n-1)$, $S(n) \geq S(n-1)$, so that 
$$S(n) - S(n-1) \geq \frac{2(n-2)}{n^2} = \frac{2}{n} - \frac{4}{n^2} \geq \frac{1}{n}$$ for $n \geq 4$, and we have that
$$S(n) \geq H_n,$$ and $Q(n) \leq \frac{1}{H_n}$.

To get a lower bound on $Q(n)$, we claim that $Q(n) \geq \frac{1}{2}Q(n-1)$; we prove this claim momentarily. Then $S(n) \leq 2S(n-1)$.  Plugging this into the equation above, we get that 
$$
S(n) - S(n-1)  = \frac{2(n-2)}{n^2}\frac{S(n)}{S(n-1)} \leq \frac{4(n-2)}{n^2}  = \frac{4}{n}-\frac{8}{n^2} \leq \frac{4}{n}.$$
Then $S(n) \leq 4H_n$, so that $Q(n) \geq 1/4H_n$.  We conclude that $Q(n) = \Theta(1/H_n) = \Theta(1/\log n)$.

To prove the claim, suppose $Q(n) < \frac{1}{2} Q(n-1)$.  Then we have that
$$Q(n-1) - \frac{2(n-2)}{n^2}Q(n-1)^2 < \frac{1}{2} Q(n-1),$$ which implies that
$$Q(n-1) < \frac{4(n-2)}{n^2}Q(n-1)^2,$$ or $$Q(n-1) > \frac{n^2}{4(n-2)},$$ which cannot hold since $Q(n) \leq 1$ for all $n \geq 2$.
\end{proof}

\begin{theorem}
	The expected running time of the algorithm is $O(n^2 \log n)$.
\end{theorem}

\begin{proof}
Let $T(n)$ be the expected running time of the algorithm on an $n$ node graph. Then
\begin{align*}
T(n) & = \frac{n-4}{n-2} T(n-1) + \frac{2}{n-2} \cdot 2T(n-1) + O(n)\\
& = \frac{1}{p_n} T(n-1) + O(n),
\end{align*}
and the analysis is essentially the same as that in Theorem \ref{thm:runtime}.
\end{proof}

\bibstyle{abbrv}
\bibliography{FPZ}

\begin{thebibliography}{1}

\bibitem{FoxPZ19}
K.~Fox, D.~Panigrahi, and F.~Zhang.
\newblock Minimum cut and minimum $k$-cut in hypergraphs via branching
  contraction.
\newblock In {\em Proceedings of the 30th Annual {ACM-SIAM} Symposium on
  Discrete Algorithms}, pages 881--896, 2019.

\bibitem{Gabow95}
H.~N. Gabow.
\newblock A matroid approach to finding edge connectivity and packing
  arborescences.
\newblock {\em Journal of Computer and System Sciences}, 50:259--273, 1995.

\bibitem{Karger93}
D.~R. Karger.
\newblock Global min-cuts in {RNC} and other ramifications of a simple min-cut
  algorithm.
\newblock In {\em Proceedings of the 4th Annual {ACM-SIAM} Symposium on
  Discrete Algorithms}, pages 21--30, 1993.

\bibitem{KargerS96}
D.~R. Karger and C.~Stein.
\newblock A new approach to the minimum cut problem.
\newblock {\em Journal of the {ACM}}, 43:601--640, 1996.

\bibitem{NagamochiI92}
H.~Nagamochi and T.~Ibaraki.
\newblock Computing edge-connectivity in multigraphs and capacitated graphs.
\newblock {\em {SIAM} Journal on Discrete Mathematics}, 5:54--66, 1992.

\end{thebibliography}

\end{document}